\documentclass[sigconf]{acmart}
\usepackage{algorithm}
\usepackage{graphicx}
\usepackage{subcaption}
\usepackage{multicol}
\usepackage{xcolor}
\graphicspath{ {./images/} }
\usepackage[noend]{algpseudocode}
\usepackage{multirow}
\setcopyright{acmcopyright}
\copyrightyear{2021}
\acmYear{2021}
\acmConference[St. Louis '21]{St. Louis '21: The International Conference for High Performance Computing, Networking, Storage, and Analysis}{November 15--19, 2021}{St. Louis, MO}
\begin{document}

\title{PAGANI: A Parallel Adaptive GPU Algorithm for Numerical Integration}

\author{Ioannis Sakiotis}
\affiliation{%
	\institution{Old Dominion University}
	\city{Norfolk}
	\state{Virginia}
	\country{USA}
}

\author{Kamesh Arumugam}
\affiliation{%
	\institution{NVIDIA}
	\city{Santa Clara}
	\state{California}
	\country{USA}
}

\author{Marc Paterno}
\affiliation{%
	\institution{Fermi National Accelerator Laboratory}
	\city{Batavia}
	\state{Illinois}
	\country{USA}
}

\author{Desh Ranjan}
\affiliation{%
	\institution{Old Dominion University}
	\city{Norfolk}
	\state{Virginia}
	\country{USA}
}

\author{Bal{\v s}a Terzi\'c}
\affiliation{%
	\institution{Old Dominion University}
	\city{Norfolk}
	\state{Virginia}
	\country{USA}
}

\author{Mohammad Zubair}
\affiliation{%
	\institution{Old Dominion University}
	\city{Norfolk}
	\state{Virginia}
	\country{USA}
}

\thanks{Work supported by the Fermi National Accelerator Laboratory, managed and operated by Fermi Research Alliance, LLC under Contract No. DE-AC02-07CH11359 with the U.S. Department of Energy. The U.S. Government retains and the publisher, by accepting the article for publication, acknowledges that the U.S. Government retains a non-exclusive, paid-up, irrevocable, world-wide license to publish or reproduce the published form of this manuscript, or allow others to do so, for U.S. Government purposes.

FERMILAB-CONF-21-081-SCD}

\thanks{This work is supported by Jefferson Science Associates, LLC, under U.S. Department of Energy (DOE) Contract No. DE-AC05-06OR23177.}


\begin{abstract}

We present a new adaptive parallel algorithm for the challenging problem of multi-dimensional numerical integration on massively parallel architectures. Adaptive algorithms have demonstrated the best performance, but efficient many-core utilization is difficult to achieve because the adaptive work-load can vary greatly across the integration space and is impossible to predict a priori.
Existing parallel algorithms utilize sequential computations on independent processors, which results in bottlenecks due to the need for data redistribution and processor synchronization.
Our algorithm employs a high-throughput approach in which all existing sub-regions are processed and sub-divided in parallel.
Repeated sub-region classification and filtering improves upon a brute-force approach and allows the algorithm to make efficient use of computation and memory resources.
A CUDA implementation shows orders of magnitude speedup over the fastest open-source CPU method and extends the achievable accuracy for difficult integrands. Our algorithm typically outperforms other existing deterministic parallel methods.

\end{abstract}


\maketitle

\section{Introduction}

The capability to perform multi-dimensional numerical integration is a recurrent need in applications across various fields such as finance, physics and computer graphics. Examples of such applications include risk management analysis and ray-tracing computations \cite{finance}\cite{jarosz08thesis}. Applications of particular interest to the authors are parameter estimation in cosmological models of galaxy clusters and simulation of beam dynamics \cite{cosmosis}\cite{beambeam}. Two types of computational methods are typically used for numerical integration - deterministic and probabilistic. The two most desired features of a good computational method are accuracy and speed. Additionally, since none of the methods will produce accurate results all the time, it is equally important that the methods provide a reasonable error-estimate on the integral estimates they compute. Most deterministic integration methods are quadrature based. Usually, the quadrature rules used to estimate the integral in a region require evaluation of the integrand on a number of points that grows exponentially with the number of dimensions. While this fact renders the deterministic algorithms unsuitable for very high dimensions, where probabilistic algorithms can still be employed to some extent, our experiments have shown that on a CPU-platform, probabilistic algorithms such as Vegas, Suave, and Divonne are consistently outperformed by a deterministic algorithm like Cuhre on a wide variety of integrals of moderate dimensions \cite{Lepage}\cite{origDCUHRE}\cite{cuba}. Furthermore, deterministic quadrature methods are of particular importance to the experimental sciences, due to their computation of error-estimates that provide some transparency and confidence in the quality of integration results.

This paper presents a deterministic parallel algorithm for multi-dimensional numerical integration suitable for GPUs that outperforms existing deterministic methods.

In principle, numerical integration can be carried out very simply - divide the integration region into ``many'' ($m$) smaller sub-regions, ``accurately'' estimate the integral in each sub-region individually ($I_i$) and simply add these estimates to get an estimate for the integral over the entire region ($\Sigma_{i=1}^{m} I_{i}$).
If we use a simple way of creating the sub-regions, e.g. via dividing each dimension into $d$ equal parts, the boundaries of the sub-regions are easy to calculate and, if $d$ is large enough, one might expect the integrand to not vary too much within each of the $d^n$ individual smaller regions hence  it easy to estimate the integral. Not only is this method simple, it is ``embarrassingly parallel'' as integral estimates for each region can be computed completely independently. Unfortunately, this approach is infeasible for higher dimensions as $d^n$ grows exponentially with $n$. For example if $n = 10$ and we need to split each dimension into $d = 20$ parts the number of sub-regions created would be $20^{10}$ which is roughly $10^{13}$. Moreover, uniform division of the integration region is not the best way to estimate the integral. The intuition is that the regions where the integrand is ``well-behaved'' do not need to be sub-divided finely to get a good estimate of the integral. Regions where it is ``ill-behaved'' (e.g. sharp peaks, many oscillations) require finer sub-division for a reliable, accurate estimate. However, when devising a general numerical integration method, we cannot assume knowledge of the integrand's behavior. Hence, we cannot split up the integration region in advance with fewer (but perhaps larger in volume) sub-regions where the integrand is ``well-behaved'' and greater number of smaller sub-regions in the region where it is ``ill-behaved''. These two reasons provide motivation for designing adaptive numerical-integration methods where the determination of the behavior of the integrand in a region and creation of sub-regions is done adaptively. A sub-region is split into smaller regions only if it is deemed that the integral estimate calculated for the sub-region is not accurate enough. 

Several adaptive integration methods have been designed and implemented in the past and are available for use (Cuba, QUADPACK, NAG, and MSL) \cite{nCubreIntegration} \cite{cuba} \cite{Berntsen1991AnAM} \cite{piessens2012quadpack} \cite{NAG}. Most of these, however, can require prohibitively long computation times for high-dimensional integrands, especially when requiring high-levels of accuracy \cite{korig}. 

As before, any adaptive integration method can itself be parallelized simply by partitioning the region into sub-regions (in some simple fashion) and applying the adaptive integration method to each sub-region in parallel. The challenge arises when the behavior of the integrand varies over the entire region and the simple partitioning results in sub-regions with very different computational requirements. This workload imbalance, as illustrated in Figure  \ref{fig:ComputationProblem}, makes effective utilization of multiple processing cores difficult to achieve, regardless of a CPU or GPU platform. This necessitates the development of efficient parallel adaptive multi-dimensional integration algorithms and implementations on emerging high-performance architectures. 

\begin{figure}[t]
	\centering
	\includegraphics[width=\linewidth]{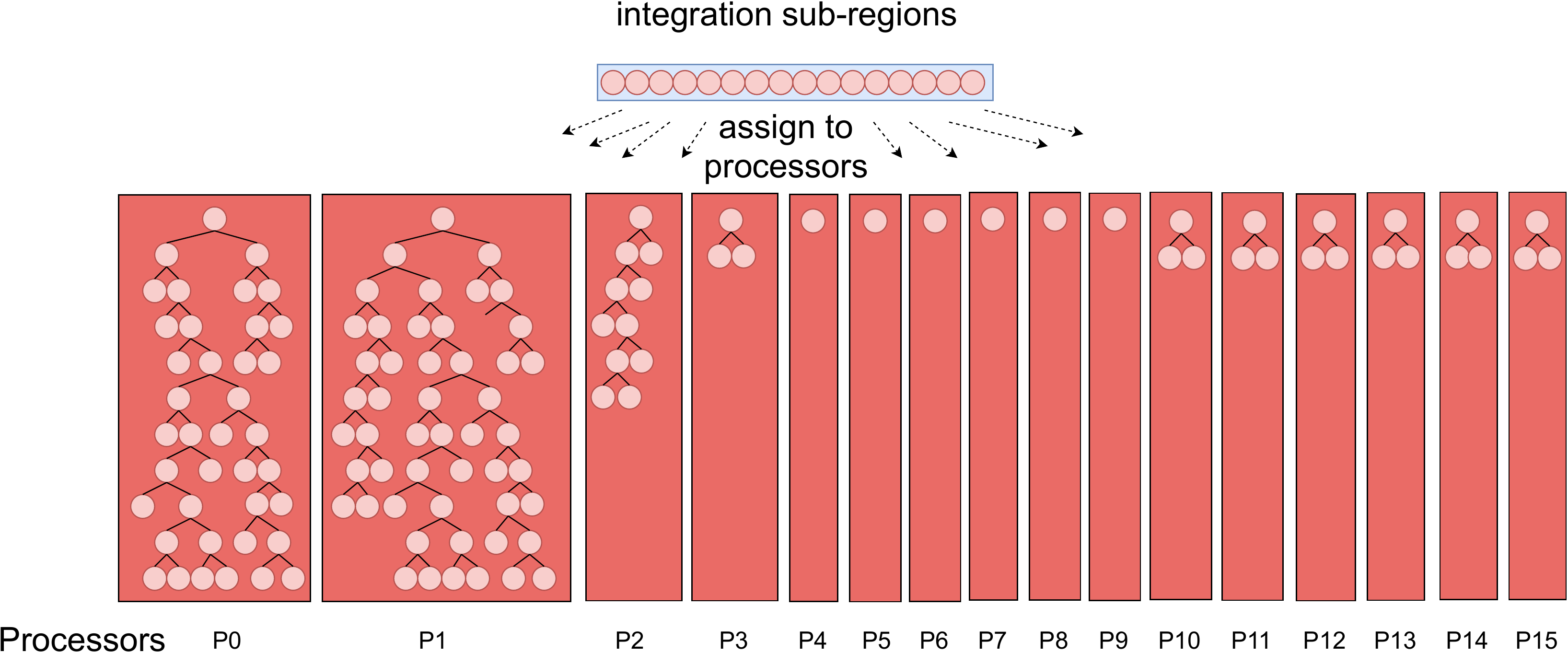}
	\caption{Parallelization of adaptive integration routines, requires the assignment of sub-regions to parallel processors. The impossible to predict work-load imbalance is evident on processors $0$ and $1$, which reach a greater depth on the sub-region tree and perform far more sub-divisions.}
	\label{fig:ComputationProblem}
\end{figure}

Parallel algorithms for multi-core platforms exist but are simple extensions of the sequential algorithms and do not address the load-balancing issues, while open-source CPU implementations expose limited parallelism \cite{dcuhre} \cite{origDCUHRE} \cite{concurrentCuba}. As a result, various algorithms for many-core systems have been developed to exploit a greater degree of parallelism in conjunction with load-balancing techniques. While these algorithms are better suited to highly-parallel architectures, they still execute a sequential adaptive integration algorithm, albeit concurrently on independent processors. The challenge with such an approach is that each processor cannot determine the achieved accuracy at the global level without synchronization or intra-processor communication. Without such information, the processors cannot decide when to terminate their local execution or predict how many resources they require. This leads to either costly synchronization/data redistribution or devoting significant computational resources (both execution time and memory use) to regions that may not contribute significantly to either the estimate of the integral or the error estimate.

We propose a new deterministic, parallel adaptive algorithm for multi-dimensional integration for massively parallel architectures. It is inspired by the Cuhre method of the Cuba library first introduced in \cite{cuba} and its parallel GPU-adaptation \cite{korig} \cite{korig2}. Unlike other parallel methods such as \cite{korig} \cite{korig2}, the proposed PAGANI algorithm does not utilize the common sequential scheme seen in adaptive integration. This avoids sequential computations in favor of greater parallelization. PAGANI further differs from \cite{korig} \cite{korig2} in some key aspects, specifically those pertaining to handling the termination condition and load-balancing. Our tests on a standard test suite of integrals show that PAGANI is at least as accurate as Cuhre and its parallel adaptations in \cite{korig} and \cite{korig2}, performs orders of magnitude faster than the serial implementation of Cuhre if the integral is not computationally trivial and can thus be used as a viable alternative on challenging integrands, especially when moderate/high digits of precision are required.

The remainder of this paper is structured as follows. In section II we briefly present related sequential and parallel methods. In section III we present the new PAGANI algorithm, while related performance findings are described in Section IV. Finally, we summarize our findings in Section V and discuss future work.

\section{Background}

There are several techniques for multi-dimensional numerical integration besides deterministic quadrature. The Cuba library provides Monte Carlo methods such as Vegas, Suave, and Divonne \cite{Lepage}\cite{cuba}. Sparse Grids methods  with adaptive and non-adaptive variants (\cite{actaNumerica}\cite{Griebel}\cite{Bungartz}\cite{Gerstner}) as well as Quasi-Monte Carlo methods (\cite{Goda2019RecentAI}\cite{NuyensQMC}\cite{MCANDQMC}) have been introduced in existing literature, with promising results. Nonetheless, the computation of error-estimates is critical for the problems targeted by Pagani. As such, we only discuss and compare libraries that provide such error-estimates, with a focus on GPU implementations.    

The adaptive quadrature algorithms discussed in this paper, typically share the control flow presented in Algorithm \ref{alg:adaptiveAlg}. Given a function $f$ and the integration bounds $b$ on $n$ dimensions, the algorithm operates on a region-list $H$ that holds the estimated values for the integral, error, and bounds of all regions. Initially, only a single region exists spanning the user-specified integration bounds $b$. As the algorithm progresses, a subset of regions from $H$ are extracted and are all split into $k$ sub-regions by sub-dividing one of the $n$ axes. The integral and error estimates as well as the axis for future splits is calculated for newly generated sub-regions and these sub-regions are then inserted back into $H$, replacing their ``parent'' regions. 
The algorithm stops when the termination condition is met; this is usually based on a maximum number of iterations/function evaluations and/or computing a sufficiently small cumulative error. 

\begin{algorithm}
    \caption{Sequential Adaptive Integration Algorithm}
     \begin{algorithmic}[1]
        \Procedure{Adaptive Numerical Integration}{$f$, $b[n]$}
            \State $H \leftarrow b$ 
             \While{(termination condition is not satisfied)}
                \State Extract a non-empty subset of regions $S$ from $H$ 
                \For{each region $R \in S $}
                    \State partition $R$ into $k$ regions along split-axis 
                    \State Let $R_1,R_2\ldots R_k$ be these $k$ regions
                     \For{$ i \leftarrow 1 \ldots k$}  
                        \State $IntEst$($R_i$) $\leftarrow$ Integral estimate for $R_i$
                        \State $ErrEst$($R_i$) $\leftarrow$ Error estimate for $R_i$
                        \State $axis_i \leftarrow$ split-axis for $R_i$  
                        \State Insert $R_i$ into $H$
                    \EndFor
                \EndFor
             \EndWhile 
            \State $GlobalIntEst \leftarrow \Sigma_{R_i \in H}\,\,IntEst(R_i)$
            \State $GlobalErrEst \leftarrow \Sigma_{R_i \in H}\,\,ErrEst(R_i)$
            \State {\bf return} ($GlobalIntEst$, $GlobalErrEst$)
        \EndProcedure
    \end{algorithmic}
    \label{alg:adaptiveAlg}
\end{algorithm}

The different algorithms in this framework differ in how the subset of regions to be further partitioned is selected in line 4, how it is split in line 7, how the integral and error estimates and axis for further split are computed in line 9, and 10 and most importantly what termination criteria is used in line 3. For brevity, we will refer to computing integral, error estimates and the split axis for a region as "evaluating" the region. Note that each region evaluation requires multiple evaluations of the integrand/function. The execution of these algorithms can be visualized as a dynamic tree where the nodes of the tree represent regions. The leaves of the tree represents the list of "active regions" ($H$). When a region $R$, represented by leaf $L$, is split into into $k$ sub-regions, we extend the tree by adding $k$ children to leaf $L$.

It is important to note that all these algorithms are heuristics. They return integral and error estimates that, while usually working well in practice, are not theoretically guaranteed to be accurate. With this in mind, our goal was to design an algorithm that is suitable for fast execution on GPUs, and that is at least as accurate as existing algorithms.

\subsection{Cuba: Cuhre}

Cuhre is the fastest deterministic open-source adaptive quadrature method, whose efficiency is attributed to its low number of required functions evaluations to compute "reasonably good" integral and error estimates for a region. The algorithm uses the Genz-Malik cubature rules to compute which axis to sub-divide, as well as the integral and error estimates \cite{GMCubatureRules}. For an $n$-dimensional region, these rules require $2^n +\Theta(n^3)$ function evaluations whereas the Gauss-Kronrod method requires $15^n$ evaluations. Cuhre closely follows Algorithm \ref{alg:adaptiveAlg}. In line 4, it selects the region in $H$ with maximum error-estimate to split each time and in line 6 it partitions it into two equal halves along the chosen axis. It terminates when the estimate of the cumulative relative-error is smaller than the user-specified relative-error tolerance or upon reaching the maximum number of pre-specified function evaluations. More precisely, the termination condition is $\frac{e}{v} \leq \tau_{rel}$, where $v$ is the global integral estimate, i.e. the sum of integral estimates of the active regions, $e$ is the global error estimate i.e. sum of error estimates of active regions, while $\tau_{rel}$ is the user-specified relative-error tolerance. Cuhre provides an error-estimate by computing the difference in the integral estimates yielded by four cubature rules of different degrees. The accuracy of this error-estimate is further improved by utilizing the two-level error-estimate computation introduced in \cite{twolevelerrorest}. 
    
\subsection{Parallel Adaptive Quadrature Methods}

Adaptive quadrature algorithms such as Cuhre, are easily parallelizable since the function evaluations required to compute the integral and error estimates within a sub-region are independent. This is exploited by Cuhre as implemented in the Cuba library, which allows the concurrent evaluation of the integrand at different points through the fork/wait POSIX functions. Adaptive methods can be parallelized further since there is no dependency amongst the sub-regions, allowing for evaluation of the individual regions independently without any need for synchronization. Parallelizing both at function and sub-region evaluation levels is important for achieving maximum speedup. Good performance is also predicated on having a balanced work-load among the processors, lacking which certain processors may operate on problematic regions while other processors with ``easier'' sub-regions could terminate early and remain idle. Such behavior can limit the benefit of parallelizing at the sub-region level and its mitigation could potentially require costly synchronization between iterations for data redistribution.

Existing parallel algorithms and their implementations vary in the set of cubature rules utilized, host/device memory utilization, data-structure for $H$, and method for load-balancing. The common characteristic across all known existing parallel methods, is that each processor eventually executes a variant of Algorithm \ref{alg:adaptiveAlg}. This has several disadvantages: (i) As noted earlier, it typically leads to poor load-balancing (ii) use of local data-structures that are accessed and maintained with local resources (memory, extent of parallelization), and (iii) the termination at individual processors is often due to reaching a maximum number of function evaluations and not due to satisfying the error-tolerance. This is necessary as any relative-error check on a local processor's sub-regions cannot guarantee anything about the global relative-error without intra-processor synchronization. Our algorithm is not of this nature and addresses those disadvantages by adopting a ``global'' approach, where the full extent of memory and parallel resources are utilized for all procedures and no locality-related restrictions apply.

\subsubsection{Two-Phase Quadrature}

The parallel algorithm introduced in \cite{korig} and refined in \cite{korig2}, operates in two phases. In phase I, the algorithm attempts to generate a sufficiently large sub-region list based on the available memory resources of the GPU. Once the size of $H$ is sufficiently large, phase II starts and each processor concurrently executes the sequential Cuhre algorithm without synchronization or any intra-processor communication. The effectiveness of this approach is attributed to the load-balancing achieved in phase I, where each region whose relative-error is no greater than the user-specified relative-error tolerance $\tau_{rel}$ is not sub-divided further (local termination) and does not participate in phase II; the algorithm discards those ``finished'' regions from memory after accumulating their contributions to the total integral and error estimates. This balances the workload since only regions with large relative-errors are expected to require a significant number of further sub-divisions, and only those computationally intense sub-regions will reach the resource-demanding phase II. 

When $\tau_{rel}$ is small (i.e. high precision is desired), the local termination condition is difficult to satisfy which often leads to all regions reaching phase II, effectively performing no workload-balancing at all. Additionally, marking regions as finished based on their relative-error can cause problems for integrands that take both negative and positive values (details in section 3.5). Finally, the algorithm utilizes the two-level error-estimate in phase II but not in phase I. This can lead to some speedup, but at the risk of over-estimating the algorithm's achieved accuracy, especially if the algorithm terminates before reaching phase II.   

\subsubsection{Heuristic Adaptive Quadrature}

A different approach was proposed in \cite{heuristicQuadrature}, where a heuristic instead of the estimated error determines the bounds of the regions that will be processed. The integrand smoothness is determined by approximating the first and second derivatives based on central differences. Large second derivatives indicate integrand variation and the need for small intervals. The heuristic yielded accurate results and demonstrated reduced number of function evaluations and improved speedup on a two-dimensional case. It is not clear if this approach remains effective on higher dimensions or in cases where the integrands have  discontinuities or sharp peaks.

\subsubsection{2D Grid of Heaps}

Another approach was presented in \cite{multicoreQuadrature}, to address the poor load balancing associated with independent processors on isolated partitions of the integration space. The processors were spatially organized on a 2D grid, where each processor operated on an isolated heap of sub-regions. Load-balancing was managed by allowing the processors to re-distribute sub-regions with their four-neighbors in the case of imbalance. Like the two-phase method, this approach is still burdened with the lack of globally aware local termination conditions and need to synchronize and redistribute sub-regions. 
    
\subsection{Quasi-Monte Carlo}

A quasi-Monte Carlo integration library was presented in \cite{qmcGPU}, allowing execution on both multi-core CPU platforms as well as GPUs through a CUDA implementation. This library merges quasi-Monte Carlo integration with importance sampling and displayed improved performance over other Monte Carlo based methods such as Vegas. Similar to Cuhre, the QMC method in this library tries to compute the integral to within a user-input relative error tolerance $\tau_{rel}$. Unlike other quasi-Monte Carlo implementations, this method returns both an integral and error estimate making it suitable for comparison with PAGANI.    
    
\section{PAGANI}

We propose a parallel adaptive algorithm that is inspired by phase I of \cite{korig} but avoids phase II and its associated problems. Our algorithm attempts to achieve greater precision and convergence rates compared to other parallel methods by iteratively sub-dividing all regions in parallel instead of the top-$k$ with the worst error-estimate. A key difference between PAGANI and existing parallel methods is that the processors are not restricted to a set of sub-regions or partition of integration-space. Instead, we perform a new 1-1 mapping between processors and regions at each iteration, and the sub-regions that we generate through sub-division are assigned to new processors at the next iteration. This eliminates the need for sequential operations such as sorting or heap-maintenance on processor-local sub-regions. Furthermore, the adaptiveness of PAGANI is not attributed to error-based sub-region sorting, which is the most common approach in methods like Cuhre. 
Instead, a filtering step at each iteration first classifies the regions that need to be processed further (``active'') and those whose further sub-division is not expected to impact the cumulative error (finished) and can be removed from memory.

\subsection{Algorithm Description}

The details of the PAGANI algorithm are described in Listing \ref{alg:bfc_alg}. The input variables are: the integrand $f$ and its dimensionality $n$, the bounds of the integration space $b$, as well as the user-specified relative and absolute error tolerances $\tau_{rel}$ and $\tau_{abs}$. The integration space is initially defined as a single region (line 2). Before proceeding with the main loop, the starting region is partitioned to $d^n$ child sub-regions by  dividing each axis into $d$ equal parts. This generates the initial list $H$, which will be iteratively expanded in the main loop through further sub-divisions. Pre-processing is finalized by initializing all required lists and variables (lines 5 to 8). The list $A$ contains integer flags to distinguish the active and finished regions. $V, E, $ and $K$ are lists that hold the integral estimates, error estimates, and next axis to split, while $V_p, E_p$ hold the estimates of the parent regions (regions processed at last iteration). We also keep the cumulative integral and error estimates of the in-memory regions in the variables $v$ and $e$, while the estimates contributed by the finished regions from all past iterations are accumulated in $v_f$ and $e_f$.

The algorithm proceeds to execute the main loop for a maximum of $it_{max}$ iterations or until its relative-error estimate reaches $\tau_{rel}$ or its  error-estimate reaches $\tau_{abs}$. Each iteration begins with the evaluation of all regions in $H$ (line 10), which returns updated lists $V, E,$ and $K$.
Then we refine the error estimates  produced by the \textsc{evaluate} method by replacing them with the two-level error-estimates. Once the finalized integral and error estimates are computed, we can classify the regions (line 12) by comparing their individual estimated relative-errors against $\tau_{rel}$. The next step is to perform a sum reduction of $V$ and $E$ (lines 13, 14), which yields the estimates of the regions being processed. 

Then the termination condition is evaluated (line 15). This involves adding the contributions of the ``leaf" regions, which consist of the current regions in $H$ and the finished regions.
If the termination condition is not met, the algorithm must determine which regions to process further. The REL-ERR-CLASSIFY method, populates the list $A$ with a $1$ or $0$ to indicate which regions will remain active for further processing. A summation on the dot product of $V$ and $A$ (line 18), accumulates the integral estimate of the active regions, and by subtracting from $v$ we can determine the estimate of the finished regions. The same process is applied for the error estimates. A secondary classification is applied (line 17) to identify regions whose relative-errors are larger than $\tau_{rel}$ but have a magnitude too small to impact the accuracy of the global estimates. A filtering operation is then applied and the finished regions are removed from all lists. The parent estimates in $V_p, E_p$ are then updated with the values of the current regions. The iteration concludes with the sub-division of all the remaining regions on the axes listed in $K$, and the setting of the active region-list size $s$, to double the size of the active regions. 

\begin{algorithm}
	\caption{PAGANI Algorithm}
	\begin{algorithmic}[1]
		\Procedure{\sc{PAGANI}}{$f$, $n$, $b[n]$, $\tau_{rel}$, $\tau_{abs}$}
		\State {$R_0 \leftarrow b$}
		\State {$s  \leftarrow d^n$}  \Comment{region list size}
		\State {$H \leftarrow $ \sc{Uniform-Split}($R_0, d$)}
		\State {$A[1:s] \leftarrow 1$}
		\State {$V[1:s], E[1:L], K[1:s] \leftarrow 0$}
		\State {$V_p[1:s], E_p[1:s] \leftarrow 0$}
		
		\State {$v, e, v_f, e_f\leftarrow 0$}   \Comment{cumulative/finished estimates}
		
		\For{\texttt{$it:it_{max}$}}
		\State {\texttt{$V, E, K \leftarrow$ \sc{Evaluate}($H$)}}
		\State {\texttt{$E \leftarrow$  \sc{Two-Level-Error}($V, E, V_p, E_p$)}}
		\State {\texttt{$A  \leftarrow$ \sc{Rel-Err-Classify}($V, E, A$)}}
		
		\State {$v \leftarrow$ \sc{Sum}($V$)}
		\State {$e \leftarrow$ \sc{Sum}($E$)}
		
		\If{$\frac{e+e_f}{\lvert v+v_f \rvert} \leq \tau_{rel}$ or $e+e_f \leq \tau_{abs}$}
		\State {\texttt{return $v+v_f, e+e_f$}}
		\EndIf
		
		\State \texttt{$A \leftarrow$ \sc{Threshold-Classify} ($A, E, v+v_f, e+e_f, v, e, s$)}
		\State {$ v_f \leftarrow$ $v -$ \sc{Sum}($V \cdot A$) $+v_f$}
		\State {$ e_f \leftarrow$ $e -$ \sc{Sum}($E \cdot A$) $+e_f$}
		\State {$H,V,E,L\leftarrow$ \sc{Filter}($H, V, E, A$)}
		\State {$V_p \leftarrow V, E_p \leftarrow E$}  \Comment{update all parents}
		\State {$H \leftarrow$ \sc{Split}($H$, $K$)}
		\State {$s \leftarrow 2s$}
		\EndFor      
		\EndProcedure
	\end{algorithmic}
	\label{alg:bfc_alg}
\end{algorithm}

The initial partitioning  of the integration space and all operations within the iteration loop can be executed in parallel. The {\sc Evaluate} method is the only place where parallelization at the function evaluation level is required, and thus a collection of  is responsible for the evaluation of a single region. The classification, filtering, and sub-division methods mainly modify the lists and perform copying operations, which allows for a single thread to process a region. Similarly, the summations can be executed through reduction operations that can exploit parallelism.

\subsection{Region Evaluation}

The success and convergence rate of sequential methods is mostly dependent on the appropriateness of the cubature rules on the target integrand. PAGANI utilizes the same cubature rules as Cuhre and can function as an equally-reliable alternative. The computation of the estimates requires a weighted summation in the form of $\sum_{i=1}^{N} w_i \cdot f(x_i)$, with the weights ($w_i$) and number of function evaluation points ($N$) and their location distinguishing the various cubature rules. 

This summation is used to derive both the integral and error estimates. A cubature rule provides a single sum for the integral estimate. To compute the error-estimate, four additional rules provide four different estimates, with the largest difference of those four yielding an error value. Similar to Cuhre, PAGANI uses this error-value to compute a more refined two-level error-estimate that is less susceptible to overestimation. While details can be found in \cite{twolevelerrorest}, the computation requires the integral estimate of the parent, as well as the integral and error estimates of the ``sibling'' sub-region. This can improve accuracy as a challenging integrand variation (e.g. very sharp peak) could be potentially exposed in the parent region but not on the children sub-regions.  


\subsection{Processor-to-Region Mapping}

The lack of the sequential Algorithm \ref{alg:adaptiveAlg} within PAGANI, allows the algorithm to map a processor to a single sub-region for the span of a single iteration. Since parallel processors perform a single split per iteration, the cumulative integral and error estimates are known after every parallel set of evaluations. This implicit synchronization at each iteration, allows the algorithm to terminate as soon as sufficiently accurate estimates are computed. This is in contrast to methods relying on the sequential adaptive algorithm, where processors are mapped to a region space and can perform many redundant sub-divisions in isolation. These isolated sub-divisions prevent knowledge of the cumulative estimates without a separate synchronization and accumulation step. The high computational cost of such repeated synchronizations is the primary reason behind existing parallel adaptive algorithms' simple but inefficient local termination conditions based on maximum number of function evaluations and/or local relative-error.

\subsection{Integration Space Sub-division}

PAGANI selects the regions to process by essentially performing a breadth-first search in the sub-region tree; the algorithm attempts to find the sub-region ``leaves'' whose error-estimate summation results in a sufficiently small relative-error. 
PAGANI is adaptive and avoids a brute force progression by discarding individual regions with sufficiently small relative-error at each iteration, essentially pruning the sub-region tree. 

Figure \ref{fig:EightIterations} illustrates how PAGANI differs from Cuhre and the two-phase algorithm of \cite{korig} in how they generate the sub-region tree. Sequential Cuhre performs one sub-division per iteration and can thus generate a narrow tree, with some tree-depths containing many regions that are never sub-divided. In contrast, the PAGANI algorithm sub-divides all regions at each depth, unless they have been classified as finished. While the two-phase algorithm follows this scheme as well, once enough sub-regions have been generated for a 1-1 mapping with all available processors, the sequential algorithm is executed by each processor on their respective sub-region. Instead, PAGANI continues the breadth-first expansion while classifying even more regions as finished by a more aggressive filtering approach. 

The sub-region trees will look different among the three methods, with the tree in PAGANI expected to be wider and potentially smaller in depth. Nonetheless, since all three methods evaluate the individual regions in the same manner, similar patterns are expected to be observed in terms of which sections of the tree will be expanded. In the example illustrated in Figure \ref{fig:EightIterations}, sequential Cuhre never expands the right node at the second level due to its small error-estimate. PAGANI and the two-phase method both sub-divide all regions at that level. At a later depth on the tree, the sub-regions in that branch are classified as finished, mirroring the region-selection of Cuhre.

\begin{figure}[t]
	\centering
	\includegraphics[width=\linewidth, scale = .8]{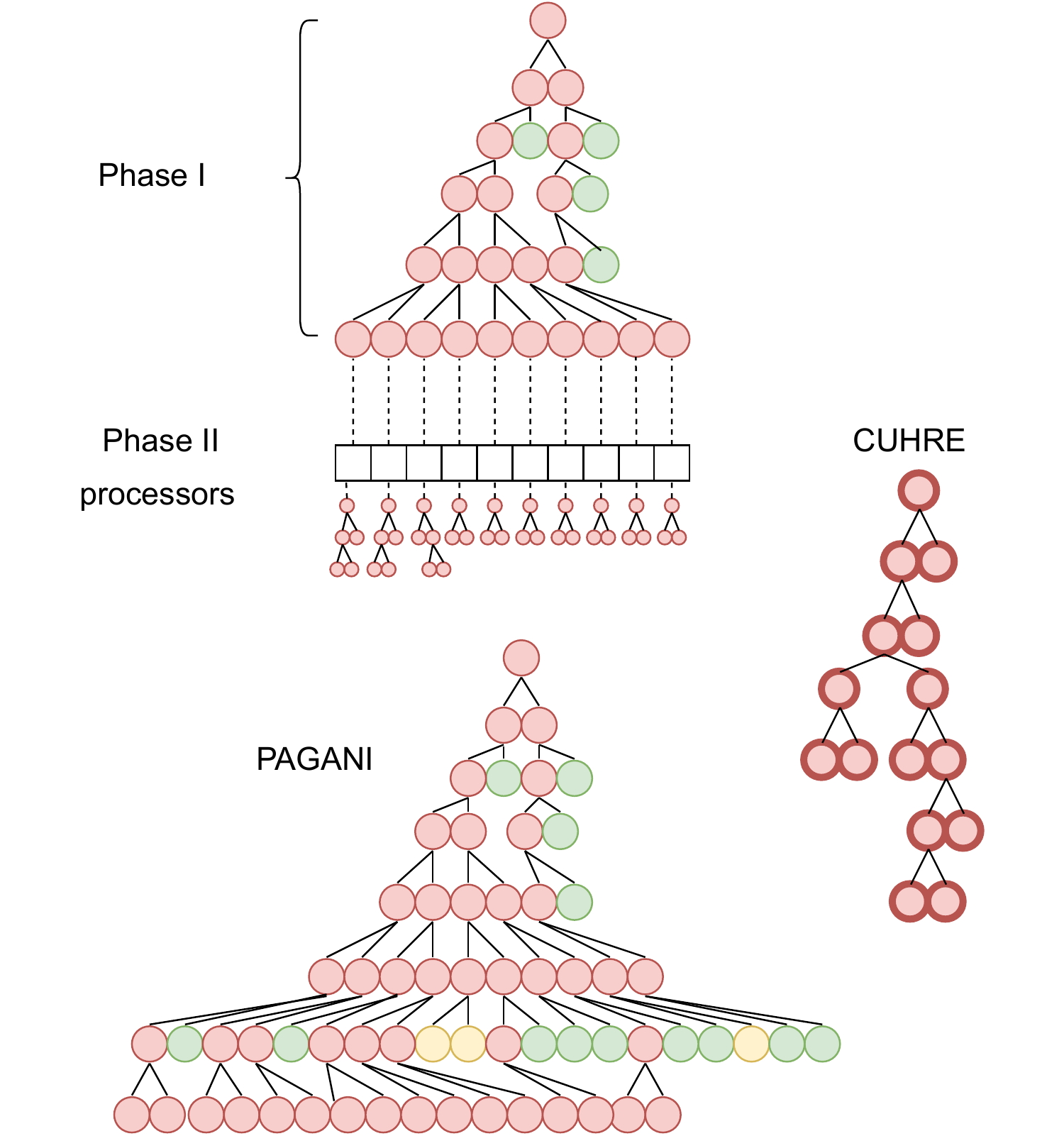}
	\caption{The state of the sub-region tree once the three adaptive quadrature algorithms complete seven iterations on an arbitrary integrand. For the two-phase algorithm and PAGANI, green nodes represent finished regions which have a relative-error smaller than $\tau_{rel}$. Note that for the two-phase algorithm, five iterations are completed in phase I and two more iterations in phase II, with some phase II processors terminating at the sixth iteration. The yellow nodes on PAGANI represent finished regions that have a relative-error larger than $\tau_{rel}$ and were classified as finished by an additional classification operation.}
	\label{fig:EightIterations}
	\centering
\end{figure}

\subsection{Adaptive Measures}\label{adaptiveMeasuresSect}

The goal of PAGANI is the same as the other adaptive algorithms, that is to identify low-contributing regions in order to better utilize computational resources and accelerate the convergence rate to sufficiently accurate estimates. Most adaptive algorithms attempt to identify these regions based on   local computations. PAGANI utilizes more sophisticated global criteria for identifying which regions are ``accurate'' or ``small'' enough to not sub-divide further without incurring any loss of accuracy. We classify such regions as finished and discard them from memory after accumulating their contributions.

\subsubsection{Relative-error filtering}

When the integrand is either always non-negative or non-positive in the entire integration domain, identifying finished regions is particularly effective due to a useful property of the estimate summations, which guarantees that if all sub-regions have a relative-error smaller than the $\tau_{rel}$, then the cumulative error satisfies the required accuracy. While not all sub-regions are expected to meet this criteria, we can stop sub-dividing the sub-regions that do, without negatively impacting the convergence rate.

\begin{lemma}
	Let $m$ denote the number of subregions. Assume that \,$\forall{i}\, 1 \leq i \leq m,\, e_i \geq 0$ and  $ v_i$'s have the same sign. Let $e = \sum_{i=1}^{m} e_i$ denote the cumulative error and $v = \sum_{i=1}^{m} v_i$ denote the cumulative integral estimate. Suppose $\forall{i} \,\, 1 \leq i \leq m,\, e_{i} \leq \lvert {v_i} \rvert \cdot \tau_{rel}.$
	Then $e \leq \lvert v \rvert \cdot \tau_{rel}$.
    
\end{lemma}

\begin{proof}
    $e = \sum_{i=1}^{m} e_i \leq \sum_{i=1}^{m} \lvert v_i \rvert \cdot \tau_{rel} \leq  \lvert \sum_{i=1}^{m}  v_i \rvert \cdot \tau_{rel}=\lvert v \rvert \cdot \tau_{rel}.$
\end{proof}

This property holds only when all integral estimates $v_i$'s have the same sign. This is usually the case for integrands that are non-negative (or non-positive) over the entire integration domain.  
In principle, for functions that can take both positive and negative values, marking sub-regions as "finished" when $e/\lvert v \rvert \leq \tau_{rel}$ could be problematic as this could potentially lead to early termination without good accuracy. In practice, we observed that the algorithm may produce inaccurate estimates only when the integrand repeatedly oscillates to both negative and positive values. For such integrals one should not discard regions based on the relative-error. If the user is aware of such behavior on the integrand, a user-configurable integer flag can be set to prevent relative-error filtering.

\subsubsection{Heuristic Threshold Classification}

Filtering based on the relative-error provides a degree of adaptiveness that is useful but not sufficient for demanding error-tolerances. When $\tau_{rel}$ is small, it is more difficult for regions to reach sufficiently small error-estimates to satisfy the relative-errors. This could lead to limited or even no relative-error filtering, possibly resulting in memory exhaustion prior to reaching a desired accuracy level. To avoid this situation, PAGANI performs an extra classification step, which identifies additional regions as finished, as long as their error-estimates are small enough to not impact the cumulative error estimate. This is performed by the \textsc{Threshold-Classify} method, which compares the error-estimates of the sub-regions against a threshold-value. Sub-regions with error-estimates smaller than the threshold, are classified as finished.

Unlike Cuhre which keeps all regions in memory, any regions that PAGANI ``filters out'' are permanently removed, with their contribution reflected only in the total finished estimates. As a result, finished error-estimates represent contributions from regions that the algorithm commits to not decreasing further. Thus, the active regions are the only ones whose error-estimates can still be decreased through sub-division, which is why the algorithm must not over-commit finished contributions by overly-aggressive filtering. 

Keeping this in mind, the \textsc{Threshold-Classify} method is only invoked in two cases for two different reasons: a) the cumulative integral estimate's log$(\frac{1}{\tau_{rel}})$ significant digits remain unchanged but the cumulative error is still too high, b) memory would be exhausted if we sub-divide all active regions. In the first case, the algorithm has essentially solved the integral, but must decrease the error further to satisfy the requested precision. Since we do not risk loss of accuracy on the integral estimate, we can safely attempt to discard regions with small error estimates. In the second case, conserving memory is the only possibility for the algorithm to continue. It is worth noting that if memory is exhausted and no additional iterations can be executed, the algorithm will return the latest integral and error-estimate with a flag pertaining to not achieving the user's accuracy requirements.

\subsubsection{Threshold Search}

Identifying a good value for threshold \textit{a priori} is not possible for arbitrary integrands. However, given an arbitrary threshold value, we can compute the number of regions that would be discarded and their contributions without actually removing them, allowing for dynamic choice of threshold. The method \textsc{Threshold-Classify} attempts various thresholds in the range defined by the minimum and maximum error estimates of the regions being processed at that specific iteration. Once, the classification yields acceptable results, \textsc{Threshold-Classify} terminates, the finished regions are filtered out, and the next iteration proceeds.

The classification step is successful if at least $50\%$ of regions are finished (conserving enough memory) and if the cumulative error of those regions does not exceed the remaining error budget; we will refer to these conditions as memory and accuracy requirements respectively. We define the error budget, as the amount by which the current cumulative error must be decreased in order to reach a relative-error of $\tau_{rel}$, which translates to successful integration. Since a pre-condition for execution of the \textsc{Threshold-Classify} method is that the integral estimate has converged to the required significant digits, we can compute a sufficient approximation as $e_b \approx e - v \cdot \tau_{rel}$. It is worth noting that if the finished error-estimate is larger than the error budget, then convergence is impossible. However, this scenario is easily detectable and hence can be avoided by choice of threshold value.

\textsc{Threshold-Classify} (Listing \ref{alg:classify_alg}) ``searches'' for a threshold by following a scheme similar to a binary search. The input variables consist of the relative-error classification list $A$ as produced by \textsc{rel-err-classify}, the error-estimate list $E$, the total integral and error estimates $e_{tot}, v_{tot}$ which include finished contributions, the number of regions being processed in the current iteration $s_{it}$ and their total error-estimate $e_{it}$.  The first threshold to try is always the average error-estimate of the currently active regions. The \textsc{threshold-apply} simply performs a comparison between the threshold and each error-estimate in $E$ and updates $A$ with the classification results. If any of the two threshold requirements (memory, accuracy) mentioned above is not satisfied, the threshold must be updated through \textsc{update-threshold} in order to be evaluated again. If fewer than 50\% of active regions are discarded, the threshold decreases, allowing more regions to surpass it. If the accuracy requirement is not satisfied the threshold increases in order to discard fewer regions. Depending on which requirement is not met, the threshold is updated at half the distance between its current value and the minimum or maximum error-estimate. This process is illustrated in Figure \ref{fig:ThresholdSearch} for a five-dimensional Gaussian integrand. Using the average threshold would remove 80\% of the active regions which more than satisfies the memory requirements, but these regions contribute to 488\% of the error-budget, which would make convergence impossible. The algorithm, repeatedly decreases the threshold which results in more favorable conditions.

\begin{algorithm}
	\caption{Threshold Classification Algorithm}
	\begin{algorithmic}[1]
		\Procedure{\sc{Threshold-Classify}}{$A, E, v_{tot}, e_{tot}, e_{it}, s_{it}$}
		\State {$P_{max} \leftarrow .25$}  \Comment{target percentage of error budget}
		\State {$e_{\bar{a}}  \leftarrow 0$} \Comment{error-estimate contribution from inactive regions}
		\State {$ s_{\bar{a}} \leftarrow 0$} \Comment{\# inactive regions}
		\State {$min, max \leftarrow MinMax(E)$}
		\State {$e_b \leftarrow e_{tot} - v_{tot} \cdot \tau_{rel}$} \Comment{error budget}
		\State {$t \leftarrow \frac{e_{it}}{s_{it}}$} \Comment{set initial threshold as avg. error-estimate}

		\Repeat
		
		\State {$A \leftarrow$ \sc{Apply-Threshold}($E, t$)}
		\State {$s_{\bar{a}} \leftarrow n -$ \textsc{sum}($A$)}
		
		\If{$s_{\bar{a}} > .5 \cdot s_{it}$} \Comment{memory requirement}
		\State {$e_{\bar{a}} \leftarrow e_{it} - $ \sc{sum}($A \cdot E$)}
		
		\If{$e_{\bar{a}} \leq P_{max} \cdot e_b$} \Comment{accuracy requirement}
		\State {\Return $A$}
		\EndIf
		\EndIf
		\State {\sc{Update-Threshold}($t, s_{\bar{a}}, e_{\bar{a}}, P_{max}, min, max$)}
		\Until{$\frac{s_{\bar{a}}}{s} > .5$ and $e_{\bar{a}} \leq P_{max} \cdot e_b$}
		
		\State {\Return $A$} \Comment{unsuccessful filtering}
		\EndProcedure
	\end{algorithmic}
	\label{alg:classify_alg}
\end{algorithm}

In order to avoid needless search in the absence of an acceptable threshold, we limit the number of times allowed to change the direction (towards min. or max. error-estimate) of the search. Furthermore, to prevent dramatic reductions to the error budget, the memory requirement is initially satisfied only if the contribution of the finished regions does not exceed 25\% of the remaining error budget. This demanding ratio is relaxed upon observing repeated direction changes. Specifically, the percentage is incremented by $10$ on every direction change, with 95\% as a maximum. These simple adjustments are performed within \textsc{update-threshold} through trivial counter variables that we do not list for brevity.

\begin{figure}[t]
	\includegraphics[width=\linewidth]{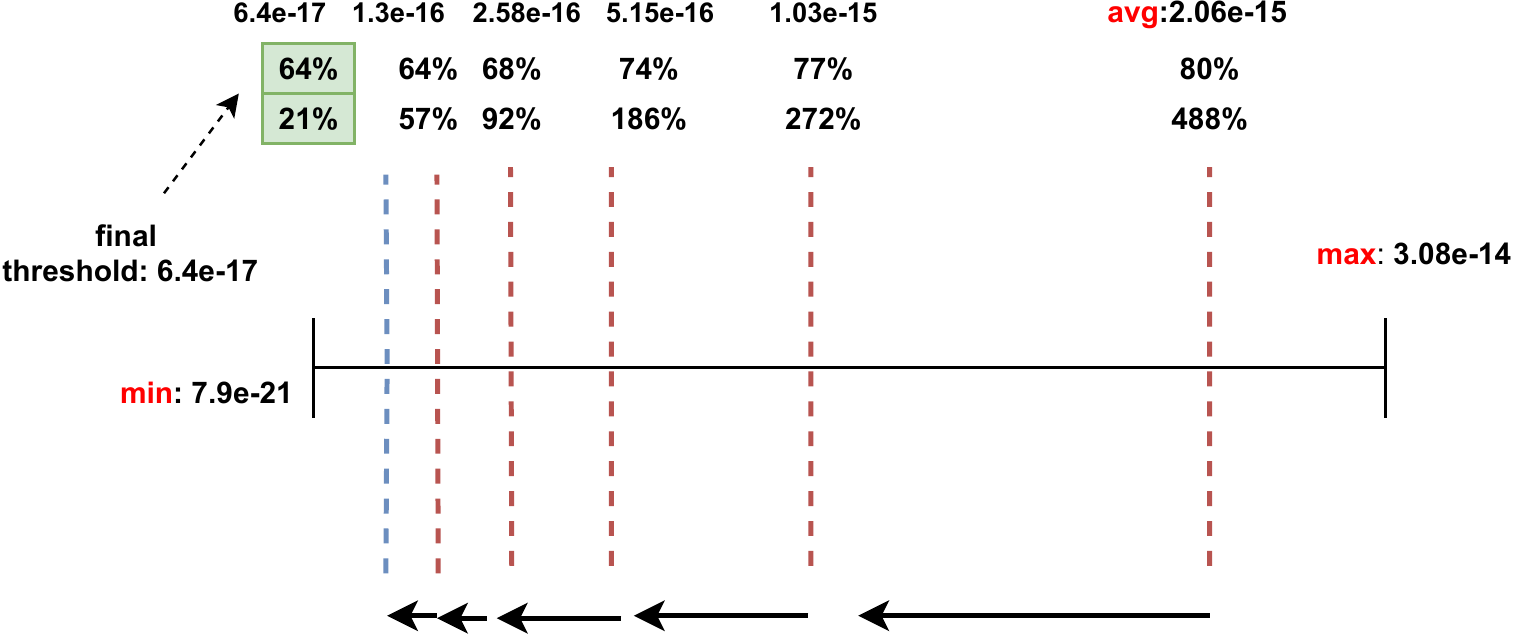}
	\caption{Threshold Search. The dotted lines represent the different values of thresholds that are attempted, with the blue line indicating the threshold that satisfied both memory and accuracy requirements. The top percentage value above the lines, indicates the percentage of current regions that would removed from memory if the candidate threshold was utilized. The bottom percentage value indicates the amount of the error budget that will be covered by the discarded regions.}
	\label{fig:ThresholdSearch}
	\centering
\end{figure}

\section{Experimental Results}

Our experiments with both parallel and sequential implementations were conducted on a compute node with a 2.4 GHz Xeon (R) Gold 6130 CPU and a 16 GB V100 GPU. The V100 provides a 16GB HBM2 on-board memory and a peak performance of 7.834 Tflops in double precision floating point arithmetic. The programming environment consists of CUDA 11, GCC 9.3.1 and utilizes c++ 17. We compare PAGANI against sequential Cuhre and the parallel two-phase method. Execution of the two-phase method was achieved with an updated implementation for a fair comparison. Execution of sequential Cuhre was conducted with the Cuba 4.0 C-implementation. Cuba allows parallelization of the function evaluations through the fork/wait POSIX functions. This method of parallelization is relatively inefficient and does not yield any speedup unless the function evaluations of the integrand  takes significant computation time. For the integrands in our test suite, no performance benefits were observed when utilizing such parallelization. To our knowledge there is no other open-source parallel version of Cuhre and as such we utilized the sequential version of Cuba. 

All experiments on Cuhre, PAGANI, and the two-phase method, were executed with the same input parameters. The absolute-error tolerance was constant and set to $10^{-20}$. Each integrand was evaluated on increasingly smaller relative-error tolerances across the range  $\tau_\textrm{rel} = (10^{-3}, 1.024 \cdot 10^{-10})$ unless failing to integrate successfully before reaching the minimum $\tau_\textrm{rel} = 1.024 \cdot 10^{-10}$. The critical parameter on all experiments is $\tau_\textrm{rel}$ as it specifies the user-required accuracy. For all plots in this section, the \textit{x-axis} represents the various $\tau_\textrm{rel}$ on which the different integrands were executed. To make the plots more intuitive we perform a simple transformation on the \textit{x-axis} with  $\log_{10}(\frac{1}{\tau_{rel}})$, which ``roughly'' translates to the number of precision-digits, e.g $\tau_\textrm{rel}$ of $10^{-3}$ translates to 3 digits of precision. The ``final'' flag in Cuhre was set to $1$, since the relative-error is otherwise prone to overestimation for sequential Cuhre. Finally, the maximum number of function evaluations on Cuhre was set to $10^{9}$ but PAGANI and the two-phase method had no such constraint as their success relies on increased computation-throughput. Timing results encompassed the integration as a single function call, which did not include work-space allocation/setup as this is only required once for the computation of various integrals. 

Pagani and the two-phase method dedicate a block with $256$ threads for each sub-region evaluation; for the two-phase method this applies to both phases. Phase I executes until reaching a maximum number of regions that can satisfy a 1-1 mapping with the parallel blocks without exhausting GPU resources; for the V100 this was set to $2^{15}$. Phase II can launch a maximum of $2\cdot 2^{15}$ blocks, with a memory space $2048$ sub-regions per block. All other kernels on both GPU methods are executed with $256$ threads per block as well. 

\subsection{Test Integrands}

We performed experiments on several functions (with known analytic integration results) derived from the families of integrands presented in \cite{cuba}. We also include two box integrals $f_7(x)$ and $f_8(x)$ in our test suite.
Each tested integrand has different characteristics such as dimensionality, shape of integration-space, presence of peaks, discontinuities, and non-differentiability, as well as box and Gaussian forms. The domain of integration for all functions is $(0,1))$ on each dimension. The plots in this section include executions of $f_1(x), f_3(x), f_4(x), f_5(x), f_7(x), f_8(x))$, in eight dimensions, $f_4(x)$ in five dimensions, $f_6(x)$ in six dimensions, and $f_3(x)$ in three dimensions.

    \begin{equation}
     f_1\left(x\right) = \cos\left(\sum_{i=1}^{8} i \, x_i\right)
    \end{equation}

	\begin{equation}
	 f_2\left(x\right) = \prod_{i=1}^{6} \left(\frac{1}{50^2} + \left(x_i - 1/2\right)^2\right)^{-1}
	\end{equation}

	\begin{equation}
	f_3\left(x\right) = \left(1+ \sum_{i=1}^{d} i \, x_i\right)^{-d-1} 
	\end{equation}

	\begin{equation}
	 f_4\left(x\right) = \exp\left(-625 \sum_{i=1}^{d} \left(x_i-1/2\right)^2\right)
	\end{equation}

	\begin{equation}
	 f_5\left(x\right) =	\exp\left(-10  \sum_{i=1}^{d} | x_i - 1/2 |\right)
	\end{equation}
	
	\begin{equation}
	 f_6\left(x\right) =	
	\begin{cases}
    \exp\left(\sum_{i=1}^6 \left(i+4\right) x_i\right) & \text{ if } x_i < \left(3+i\right)/10 \\
    0                           & \text{otherwise}
	\end{cases}
	\end{equation}
   
   \begin{equation}
   	f_7\left(x\right) =  \left(\sum_{i=1}^{d} x_i^2\right)^{11}
   \end{equation}
   
    \begin{equation}
   	f_8\left(x\right) =  \left(\sum_{i=1}^{d} x_i^2\right)^{15/2}
   \end{equation}
\subsection{Accuracy}
    
Cubature based algorithms like Cuhre and PAGANI report convergence to sufficiently accurate results when the ``estimated'' relative-error is smaller than $\tau_\textrm{rel}$ or the error-estimate is smaller than $\tau_\textrm{abs}$. The estimates of the algorithms are not expected to be identical to the true relative/absolute error and to our knowledge, no algorithm can guarantee such claims. Past  work has demonstrated that the two-level error-estimation based on the differences of various cubature rules, computes sufficiently accurate results in the vast majority of cases; in our experiments only the case of the $5D$ $f_5(x)$ integrand, caused Cuhre to overestimate the error and this behavior was only observed after 7-digits of precision. Nonetheless, we believe that any analysis of numerical integration algorithms, whether sequential or parallel, deterministic or probabilistic, should evaluate the accuracy of the yielded estimates, especially when investigating performance. 

The standard testing approach proposed in \cite{GenzTest} involves comparing the execution times of integrands which consist of randomized parameters that specify the location/difficulty of integrand features such as peaks and discontinuities. The different characteristics of the test integrand families are critical for evaluating the robustness and adaptability of an algorithm under different circumstances. As such, this testing approach is well suited for investigating execution times but does not facilitate inspection of the achieved accuracy. The randomness of the parameters makes it challenging to determine an analytic evaluation of each integrand, making any comparison between the true values and the estimates difficult to achieve. To solve this problem and avoid algorithms that execute fast but with sub-par accuracy, we propose the utilization of the same integrand families presented in \cite{GenzTest}, but with fixed parameters that allow analytic results to be computed and standardized for further testing. This will allow comparison between the estimated relative-error and true relative-error, both of which typically need be smaller or equal to $\tau_{rel}$ to satisfy the user's desired accuracy. 
    
Both PAGANI and the two-phase algorithm match Cuhre in digits-of-precision, with a few exceptions such as 6D $f_6(x)$ and 5D $f_4(x)$ where the two-phase method reached two fewer digits-of-precision compared to Cuhre; this deficiency is attributed to the weaker load-balancing on high-precision (small $\tau_{rel}$) runs. On integrands that were challenging for the sequential method, PAGANI consistently outperformed both Cuhre and the two-phase method by several digits-of-precision. Figure \ref{fig:accuracy} demonstrates this behavior. The $f_7(x)$ integrand is less challenging despite the high number of dimensions and both parallel methods achieve similar levels of precision. On the contrary, $f_4(x)$ and $f_6(x)$ which consist of low-contributions around a Gaussian and discontinuities respectively, expose the fundamental computation problems of parallel adaptive integration; the two-phase method's poor load-balancing leads to the early exhaustion of the allocated memory resources, and failure to compute higher precision results. For example, the two-phase method does not achieve 6-digits of precision for 5D $f_4(x)$, as the algorithm does not assign enough processors to the integration-space where the integrand is ``bell-curved''. The PAGANI relative-error filtering (see section \ref{adaptiveMeasuresSect}) shares this deficiency and repeatedly sub-divides those low-contributing sub-regions as well, nearly exhausting memory after 20 iterations. At that point, the heuristic-search filtering is triggered and easily identifies those sub-regions as finished, retaining only $.004$\% of the sub-regions and allowing quick convergence after 7 additional iterations. The high throughput sub-division approach in PAGANI coupled with the filtering measures, consistently yield greater reliability in terms of precision when compared against Cuhre and the two-phase method without succumbing to long execution times. 

\begin{figure}[t]
	\centering
	\includegraphics[width=\linewidth]{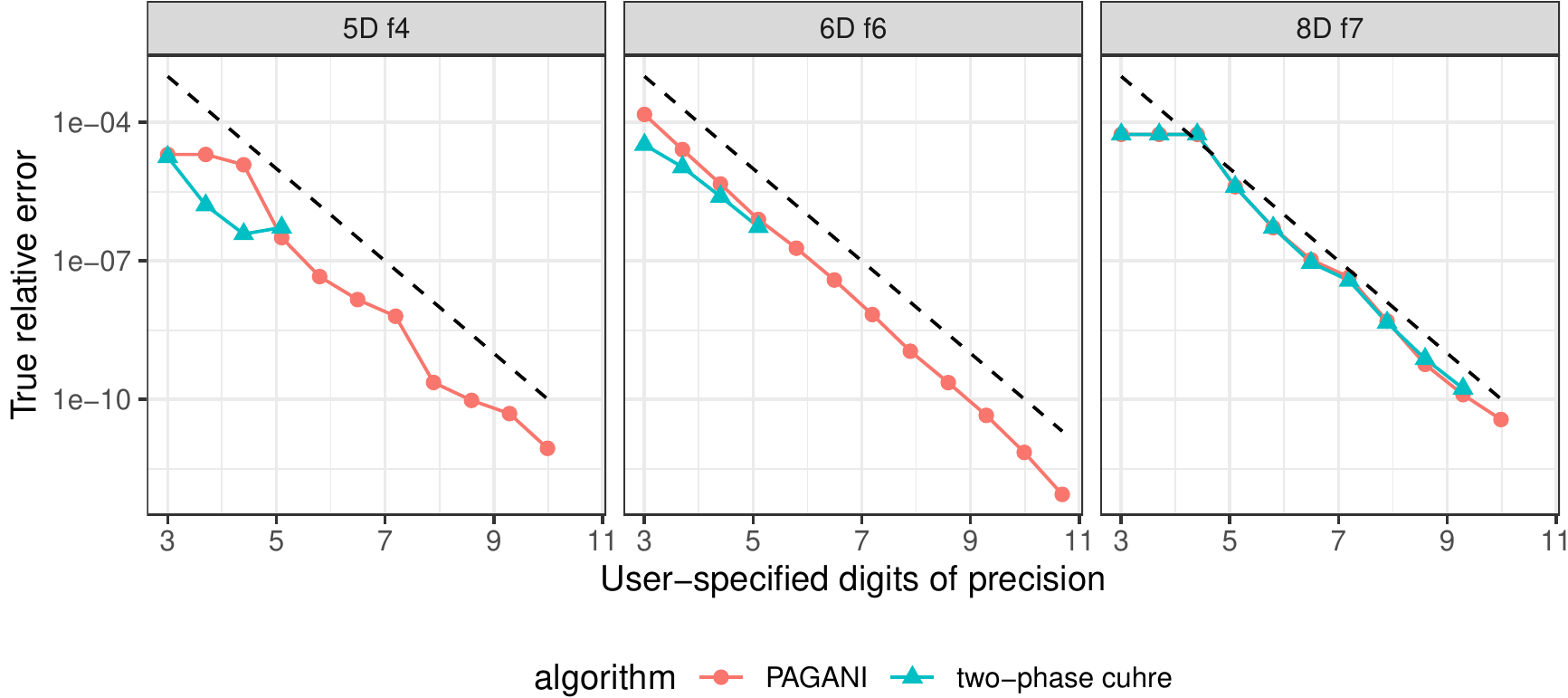}
	\caption{The dotted black line represents the user-specified relative-error tolerance across the experiments. The colored dots represent the true relative-error and their placement below the line indicates that the estimates of the algorithm were close enough to the true values to satisfy the user's desired error tolerance. Placement above the line indicates the overestimation of the attained accuracy at the moment of termination. The two-phase method fails to integrate 5D $f_4(x)$ and 6D $f_6(x)$ for more than $5$ digits-of-precision as poor load-balancing leads to the exhaustion of allocated GPU memory.}
	\label{fig:accuracy}
\end{figure}

\subsection{Performance}
 
The parallel quadrature methods consistently execute orders of magnitude faster than sequential Cuhre if the evaluation of the integrand is not trivial and requires at least a few thousand sub-divisions. When the integrand requires a small number of sub-divisions, the overhead of transferring the sub-region data and allocating resources on the GPU outweighs the benefits of parallelization. Low dimensional integrands such as $f_3(x)$ and even 5D $f_4(x)$ on low precision runs displayed such behavior. The typical speedup for challenging integrands is illustrated in Figure \ref{fig:speedupCombined}. The low precision runs ($3$-$6$ digits of precision) are as little as $15$ times faster on PAGANI but increasing precision yields speedup in the order of thousands. 

When comparing the two-phase method and Pagani, the speedup is less dramatic. They perform similarly when the two-phase method spends little or zero computation time in phase II, since phase I expands the sub-region tree in an identical manner to PAGANI. Figure \ref{fig:speedupCombined} demonstrates this behavior; the execution times are very similar on 8D $f_7(x)$ and 5D $f_5(x)$ for the first 6 digits of precision, where phase II is invoked for very limited sub-divisions or not at all. The 6D $f_6(x)$ integrand makes heavy use of phase II even for three digits of precision and yields some speedup. While not as dramatic as the speedup over Cuhre, PAGANI demonstrates some speedup on moderate digits of precision. There were cases where the two-phase method was faster, such as for 8D $f_4(x)$ and 8D $f_3(x)$, but these cases were consistently coupled with early memory exhaustion for the two-phaes method and failure to integrate for precision levels that were viable for PAGANI. The last points plotted for each integrand on Figure \ref{fig:speedupCombined} shows that the two-phase method cannot converge while PAGANI does. This is further demonstrated in Figure \ref{fig:FinalTimeComparisons} that compares all three methods. PAGANI consistently reaches higher precisions in contrast to the sequential and two-phase method. 

The quasi-Monte Carlo (QMC) method of \cite{qmcGPU}, as executed on GPU displayed stable behavior with the error estimates satisfying the user-specified $\tau_{rel}$. This applied to all functions in our test suite, except 5D $f_4(x)$ which was not correctly evaluated for three digits-of-precision. The QMC method surpassed PAGANI in the attainable precision on 8D $f_1(x)$, reaching nine digits-of-precision, while PAGANI failed at five digits. This is a result of the integrand oscillating with both negative and positive contributions, which makes relative-error filtering unreliable and is turned off via a user inputted flag. PAGANI consistently outperformed the QMC method on the remaining integrals as illustrated in Figure \ref{fig:speedupOverQMC}. 

\begin{figure}[t]
	\centering
	\includegraphics[width=\linewidth]{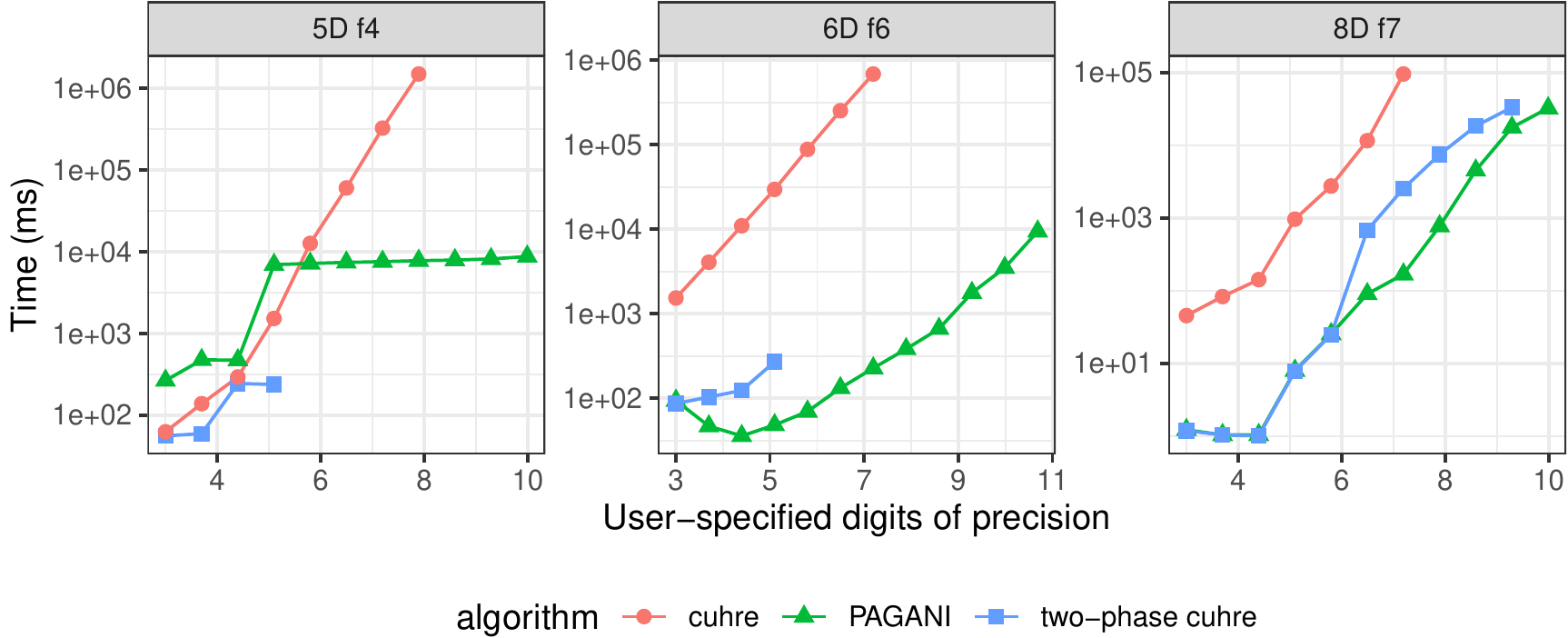}
	\caption{Execution time comparison.}
	\label{fig:FinalTimeComparisons}
\end{figure}

\begin{figure}[t]
	\centering
	\includegraphics[width=\linewidth]{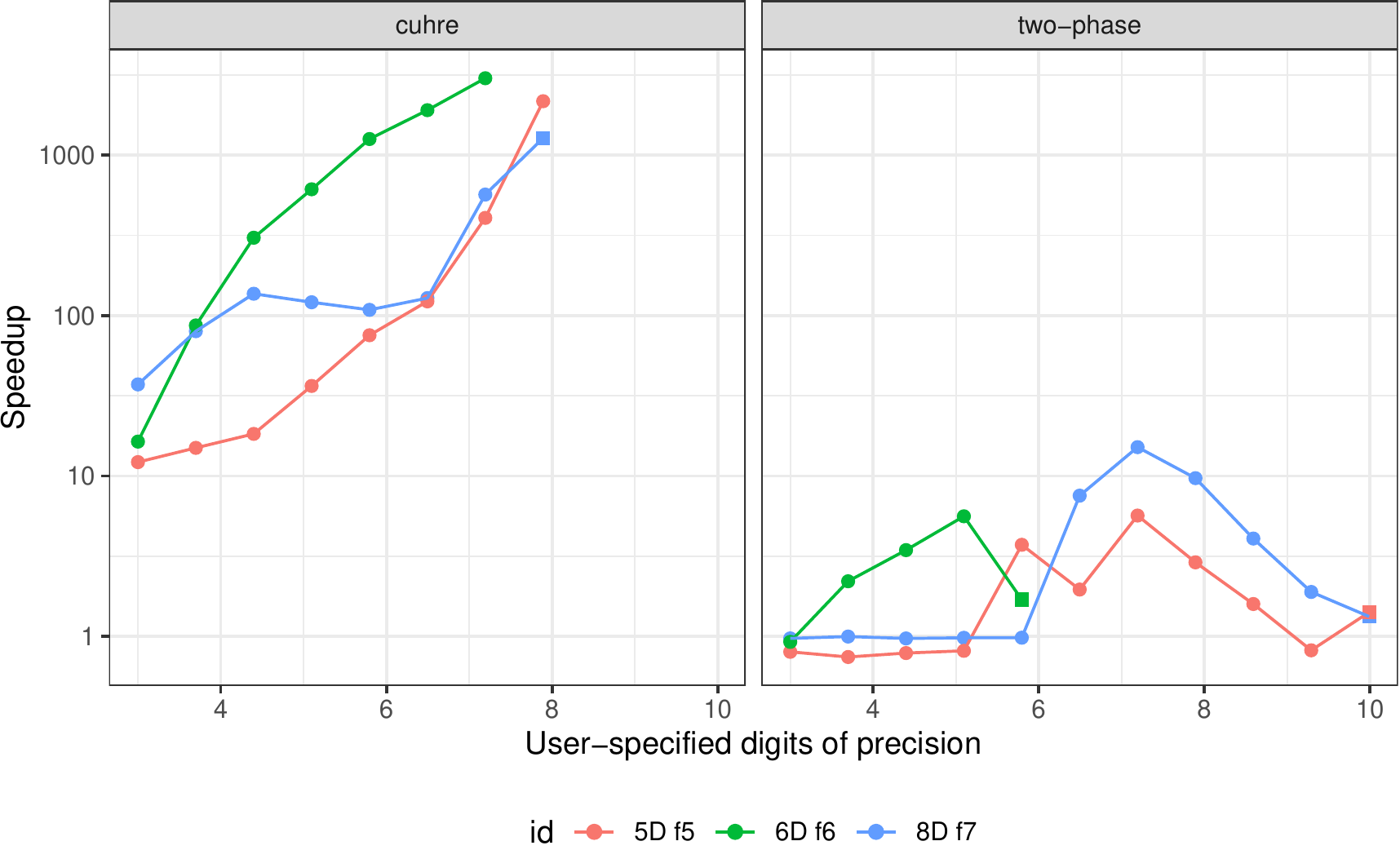}
	\caption{Speedup over Cuhre (left), two-phase Cuhre (right). Squares indicate that only PAGANI satisfied the desired accuracy.}
	\label{fig:speedupCombined}
\end{figure}

\begin{figure}[t]
	\centering
	\includegraphics[ scale = 0.4]{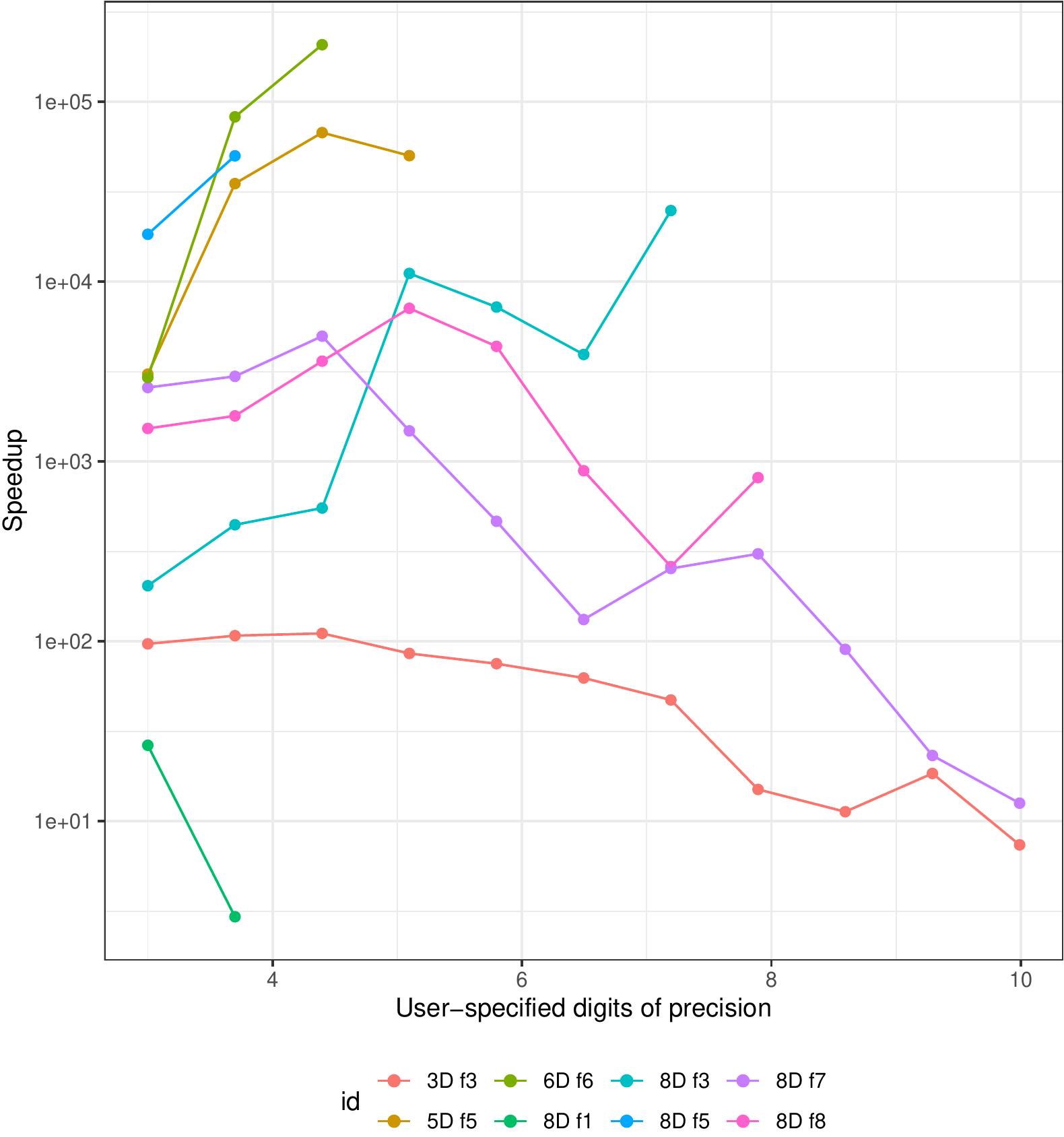}
	\caption{Speedup of PAGANI over quasi-Monte Carlo GPU implementation. Pagani is orders of magnitude faster, but is outperformed on 8D $f_1(x)$ in terms of attainable precision.}
	\label{fig:speedupOverQMC}
\end{figure}

\subsubsection{Filtering Approach}
 
The convergence rate of the PAGANI algorithm greatly benefits from the larger region throughput. The per iteration synchronization means that the extra regions do not negatively affect the sub-division process, as long as memory is not exhausted. The memory requirements increase dramatically without sufficient region filtering. Thus, the additional search-based filtering is critical in order to avoid the load-balancing issues observed in phase II of the two-phase algorithm. Filtering when memory is about to be exhausted can improve execution time (when compared to offloading to CPU memory). Further filtering when the significant digits-of-precision in integral estimate remain unchanged, leads to more significant speedup (Figure \ref{fig:FilteringsComp}), as computation focuses on the regions with greater contribution much earlier without any loss on precision.

\begin{figure}[t]
	\centering
	\includegraphics[width=\linewidth]{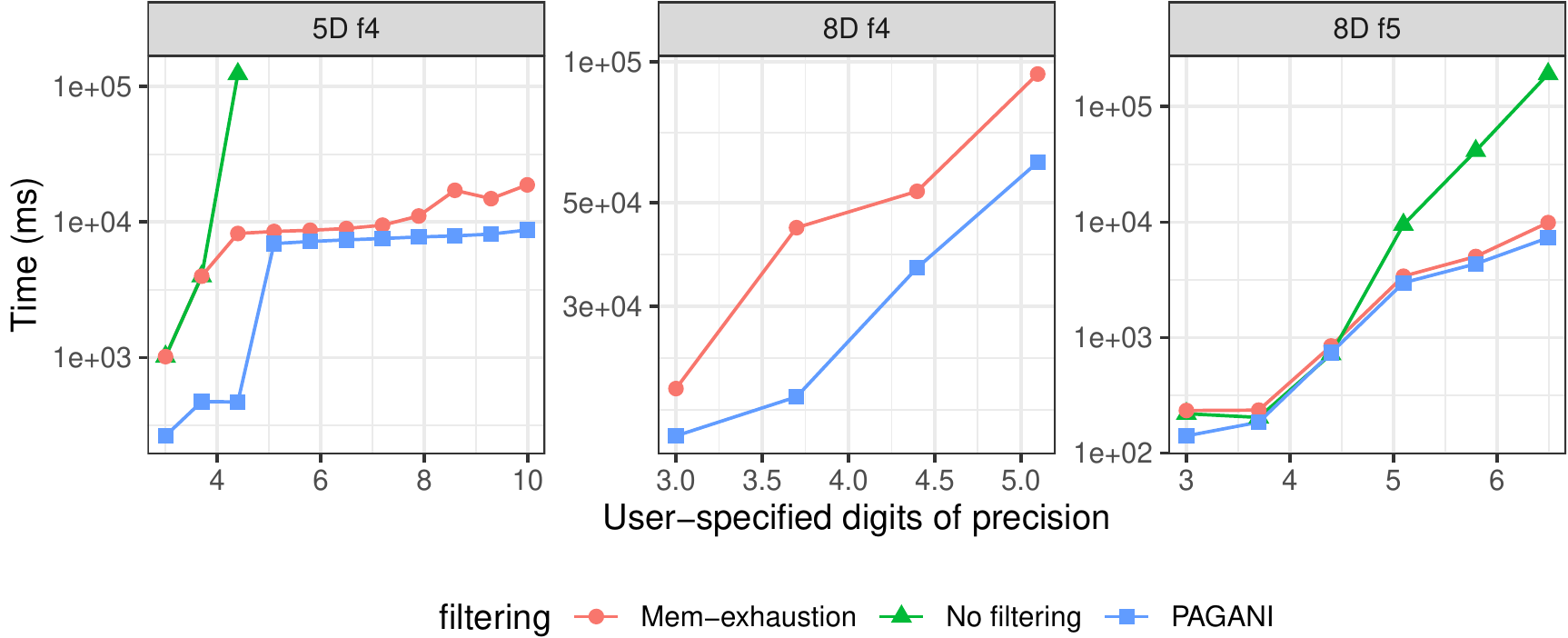}
	\caption{Execution times in milliseconds of PAGANI with/without filtering. On 8D $f_4(x)$ PAGANI without any heuristic filtering, cannot converge even at 3-digits of precision.}
	\label{fig:FilteringsComp}
\end{figure}

\subsubsection{Performance Breakdown}
Each PAGANI iteration consists of four main categories of computation: a) evaluating the sub-regions, b) post-processing that accumulates the results and performs relative-error classification c) threshold classification, and d) filtering and sub-dividing active regions.

More than 90\% of the execution time is consistently dominated by the \textsc{evaluate} method. It processes all sub-regions in parallel, computes the estimates, and determines the axis for the next sub-division. The \textsc{evaluate} method consists of a single GPU kernel and is compute bound with a limited number of memory accesses that are serviced by either shared memory or the read-only cache. The \textsc{evaluate} kernel requires a workload of at least $2^{11}$ sub-regions to reach $40\%$ (maximum at $45\%$) peak-performance on double-precision operations on a V100 GPU. Initial iterations with fewer regions lack sufficient workload for the GPU, leading to fast execution times but inefficient resource utilization. As the number of active regions is doubled at each iteration, Pagani spends few iterations on small sub-region lists. Depending on the integrand's dimensionality which determines the initial sub-region list size, this requires less than seven iterations to generate sufficiently large workloads.

As soon as the sub-regions are evaluated, post-processing operations are performed, computing the two-level error-estimate with a dedicated GPU kernel and then accumulating the integral and error estimates with the reduction methods of the Thrust library \cite{thrustLibrary}. This step also includes the relative-error classification. These operations are simple, with coalesced memory accesses and few mathematical and conditional computations. Depending on convergence of the integral estimate and/or memory exhaustion, heuristic-search based classification occurs. This process is not computationally intense and mainly involves repeated reductions and pre-fix scans which are facilitated by Thrust. 

While not comparable to \textsc{evaluate}, the filtering and sub-division computations are consistently more costly than post-processing and sub-region classification. This is caused by memory allocations required when filtering out the finished regions, and two GPU kernels that are responsible for transferring the active regions' attributes to new memory locations and sub-dividing the active regions. 

\subsubsection{Computational Requirements}

The breadth-first expansion of the sub-region tree is generally expected to generate a larger number of regions than any method that relies on a sequential algorithm. Concurrent execution of the serial algorithm on parallel processors limits the sub-region parallelization, but retains concurrency on the function evaluations. This approach generates regions by always choosing a single branch (one for each processor) to expand. Such methods can still generate a larger than necessary number of sub-regions when the precision is high, due to inefficient local termination conditions. In contrast, PAGANI exponentially increases the number of sub-regions most of which are usually filtered out at later iterations. We observe that in the cases where the two-phase method manages to integrate with comparable number of precision-digits as PAGANI, their region requirements are very similar and PAGANI shows moderate speedup. In the cases where PAGANI is slower, in some cases requiring even 100 times more sub-regions, the two-phase method would consistently fail on moderate number of precision digits. 

Figure \ref{fig:RegionsComparisons} illustrates the number of regions generated by the three integration methods. The case of 8D $f_7(x)$ displays stable behavior on both parallel methods though PAGANI suceeds for an additional digit of precision. The speedup on this integral is purely a result of the breadth-first sub-divisions scheme as the heuristic search classification is only invoked at ten digits of precision which results in removing 56\% of the active regions. For the 5D $f_4(x)$ integrand, the heuristic search classification is the main factor behind the performance improvement when integrating for more than 5-digits of precision. The two-phase method and even Cuhre initially require much fewer regions and result in faster execution times (see Figure \ref{fig:FinalTimeComparisons}). For the 6D $f_6(x)$, which mostly consists of discontinuities, we see that the heuristic classification is more aggressive than necessary on low-precision runs, as is evident by the boost on the larger-precision runs (4 digits of precision). Still, PAGANI manages to adapt through the heuristic classification and attains 10 digits of precision while the two-phase method and Cuhre fail much earlier.

\begin{figure}[t]
	\centering
	\includegraphics[width=\linewidth]{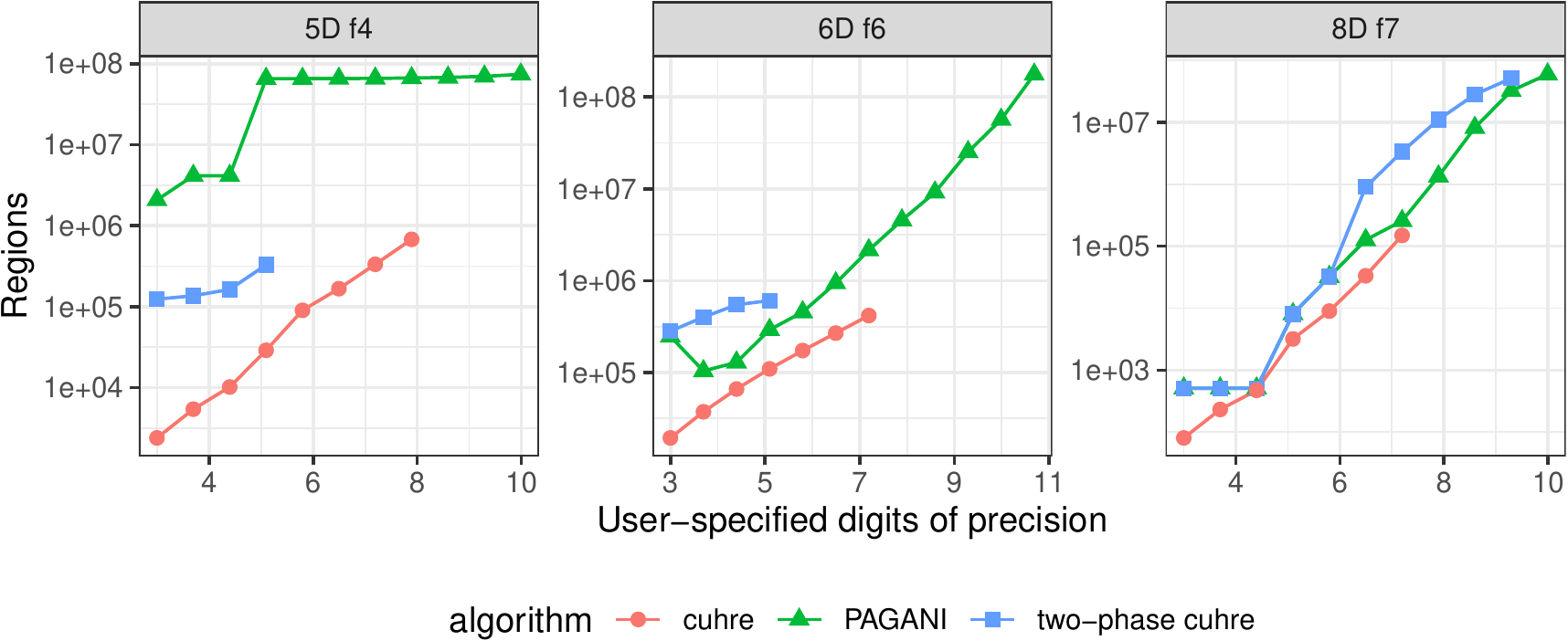}
	\caption{Generated number of sub-regions.}
	\label{fig:RegionsComparisons}
\end{figure}

\subsection{Future Multi-GPU Implementation}

The limiting factor that prevents Pagani from achieving greater precision on certain challenging integrands such as the eight-dimensional $f_5(x)$ and $f_1(x)$ is the available GPU memory. To extend memory resources, we can utilize multiple GPUs to evaluate different partitions of the integration space independently. This strategy yielded promising results in\cite{korig2} where phase I was executed on a single GPU and the generated regions were distributed to multiple GPUs for phase II.
Such an approach, would still be susceptible to load balancing issues. Furthermore, dynamic distribution of the sub-region list to different GPUs through MPI at each iteration is expected to be infeasible. As such, region redistribution would most likely be beneficial either at the beginning of the algorithm, after a set-number of sub-regions is generated, or when GPU memory is exhausted.

\section{Conclusion}

We presented a new Breadth-First cubature algorithm for deterministic adaptive multi-dimensional integration. We used sequential Cuhre and a parallel two-phase GPU adaptation of Cuhre as the basis for comparison. PAGANI showed orders-of-magnitude speedup over sequential Cuhre when the integrand was computationally challenging, whether the function evaluations were costly or the number of required sub-divisions too large. Compared to the two-phase method, the execution time of PAGANI is expected to be comparable on low precision, with some speedup ($4$ to $15$ times faster) observed on moderate digits-of-precision. The critical improvement over the two-phase method is in terms of robustness. PAGANI successfully integrated for up to 10-digits of precision while the two-phase method consistently failed much earlier. PAGANI is expected to be appropriate in all cases where Cuhre and two-phase method are utilized successfully, and most importantly when they fail due to resource exhaustion. 

The performance benefits of PAGANI are predicated on highly-parallel architectures. Increased memory capabilities of emerging architectures or potential multi-GPU execution would increase robustness further. The CUDA implementation we utilized in our experiments is a prototype, with multi-GPU execution and optimization of the GPU kernels for the Volta/Ampere architecture planned for future work. We will also be implementing the PAGANI algorithm in the Kokkos programming model, which would allow increased portability on various GPU architectures. 

\bibliographystyle{unsrt}
\bibliography{bibliography/refs.bib}
\nocite{*} 

\end{document}